\documentclass[12pt, draftclsnofoot, onecolumn]{IEEEtran}
\usepackage{ifpdf}
\usepackage{cite}
\usepackage{algorithmic}
\usepackage{algorithm}
\ifCLASSINFOpdf
   \usepackage[pdftex]{graphicx}

\else
 \usepackage[dvips]{graphicx}
\fi
\usepackage[cmex10]{amsmath}
\usepackage{epstopdf}
\usepackage{array}
\usepackage{flushend}
\usepackage{balance}
\usepackage{eqparbox}

\usepackage[tight,footnotesize]{subfigure}
\usepackage{fixltx2e}
\usepackage{color}
\usepackage{float}

\usepackage{dblfloatfix}
\usepackage{url}
\usepackage{cite}
\newtheorem{proposition}{Proposition}
\newtheorem{proof}{Proof}
\newtheorem{remark}{Remark}

\usepackage{amssymb}
\begin{document}

\title{User Pairing, Link Selection and Power Allocation for Cooperative NOMA Hybrid VLC/RF Systems }

\setlength{\columnsep}{0.21 in}

\author{Mohanad~Obeed,~\IEEEmembership{Student Member,~IEEE,}
        Hayssam~Dahrouj,~\IEEEmembership{Senior Member,~IEEE,}
        Anas~M.~Salhab,~\IEEEmembership{Senior Member,~IEEE,}
        Salam~A.~Zummo,~\IEEEmembership{Senior Member,~IEEE,} and
        Mohamed-Slim~Alouini,~\IEEEmembership{Fellow,~IEEE}}
\maketitle

\begin{abstract}

Despite the promising high-data rate features of visible light communications (VLC), they still suffer from unbalanced services due to blockages and channel fluctuation among users. This paper introduces and evaluates a new transmission scheme which adopts cooperative non-orthogonal multiple access (Co-NOMA) in hybrid VLC/radio-frequency (RF) systems, so as to improve both system sum-rate and fairness. Consider a network consisting of one VLC access point (AP) and multiple strong and weak users, where each weak user is paired with a strong user. Each weak user can be served either directly by the VLC AP, or via the strong user which converts light information received through the VLC link, and forwards the information to the weak user via the RF link. The paper then maximizes a network-wide weighted sum-rate, so as to jointly determine the strong-weak user-pairs, the serving link of each weak user (i.e., either direct VLC or hybrid VLC/RF), and the power of each user message, subject to user connectivity and transmit power constraints. The paper tackles such a mixed-integer non-convex optimization problem using an iterative approach. Simulations show that the proposed scheme significantly improves the VLC network performance (i.e., sum-rate and fairness) as compared to the conventional NOMA scheme.

\end{abstract}
\IEEEpeerreviewmaketitle
\section{Introduction}
\subsection{Overview}
Toward meeting the escalating data rates demand, visible light provides a valuable, useful  spectrum for transmitting data, which complements conventional radio-frequency (RF) communication systems. As the energy-efficient light emitting diodes (LEDs) have become more popular as light sources in indoor and outdoor environments, visible light communication (VLC) has emerged as a promising energy-efficient solution for transmitting information using LEDs. It has been proven that VLC networks can provide data rates of several Giga-bits per second \cite{schrenk2018visible, tsonev2015towards}, which makes them a powerful alternative (or complementary)  to RF networks. However,  multi-user VLC networks still suffer from the limited coverage, the channel quality fluctuation among users, and the blockages effect on the VLC link, resulting in unbalanced services, whereby some users with strong channel gains (hereafter denoted as strong users) are well served, and other users with weak channel gains (hereafter denoted as weak users) are poorly served. The communication in VLC networks is also a strong function of the existence of line-of-sight (LoS) links, which in turn get significantly attenuated with distance between the transmitter and the receiver, thereby limiting the coverage area to tiny cells, also known as attocells \cite{Yin2018}. Another challenging parameter in characterizing VLC networks is the users' distinct field-of-views (FoVs), which affect the users coverage and channels quality, and lead to unbalanced performance among the different users \cite{Obeed2018}.



In an effort to increase the throughput and improve the
fairness in wireless systems, non-orthogonal multiple access (NOMA) technique was introduced in the recent wireless literature. NOMA is based on sharing the resource components (subcarrier, spreading code, or time slot) by more than one user. This can be implemented by assigning a low power for the strong users, and a high power for the weak ones.
The weak users then decode their own messages and treat the strong users' messages as interference, while the strong users first decode the weak users' messages, remove it from the total received signal, and finally detect their message. This scheme improves the system performance in VLC networks; however, it cannot extend the system coverage or mitigate the blocking effect.

In RF networks, a recently proposed alternative to NOMA is cooperative NOMA (Co-NOMA), specifically proposed to improve the fairness and strengthen the signal-to-noise ratio (SNR) at the weak users \cite{Liu2016}, by exploiting any redundant information in NOMA. This could be implemented in practical networks, where the strong user can also work as a relay to assist the weak user. The weak user then combines both signals coming from the transmitter and from the the strong user.  This technique (Co-NOMA) has not been applied in VLC networks in the past literature, mainly because of the physical challenges incurred while forwarding
the light signal by the strong user to the weak user. One possible technique to overcome such challenges, however, is to convert the received light signal at the strong user into an RF signal, and then forward it to the weak user through an RF link. This paper adopts such Co-NOMA scheme in an hybrid VLC/RF system, and thoroughly illustrates its capabilities in enhancing the sum-rate and fairness as compared to the conventional NOMA scheme.


Consider a VLC network consisting of one AP and two predefined sets of strong and weak users. Each weak user can be paired by one strong user through an RF link. The strong user is served directly by the VLC link. Each weak user can be either served by the VLC link, or by the strong user through the RF link by means of Co-NOMA. The performance of the system becomes, therefore, a function of the user pairing, link selection and the messages powers.
The paper then focuses on maximizing a network-wide weighted sum-rate, so as to jointly pair the strong and weak users, allocate their respective powers, and select the mode of operation of each weak user (i.e., direct VLC or hybrid VLC/RF).

\subsection{Related Work} 

This subsection presents an overview of the recent state-of-art in VLC systems, with a special focus on the works which mitigate SNR fluctuations, manage NOMA networks, and analyse CO-NOMA in the realm of VLC systems.

Different techniques have been investigated in the literature to mitigate the SNR fluctuations so as to increase the system coverage probability, and improve the system performance in terms of both total achievable data rate and system fairness  \cite{coop,hand,wang2017optimization,wu2017access,JOCN, globecom,hanzo_haas,TWC,chen2013joint, chen2017performance}. Among the common solutions to strike such a trade-off between throughput and fairness is deploying hybrid VLC/RF networks \cite{hand,wang2017optimization,wu2017access,JOCN, globecom}, coordinating transmissions \cite{coop, hanzo_haas,TWC,chen2013joint, chen2017performance}, and supporting the network with relay-assisted VLC transmission \cite{dual1, dual2}. Hybrid VLC/RF technology, in particular, helps supporting the VLC standalone systems by one or multiple RF APs. The main idea in such systems is to compromise between the high VLC capacity and high RF coverage, which is often realized by assigning the users which suffer from interference, blockages, frequent handover, or low-quality channel in VLC network to be served by  RF AP(s), while serving the rest of users by the VLC AP(s)\cite{Obeed2018}. In the same direction, the authors in \cite{coop} tackle the problem of jointly optimizing the time slots and assigning the APs to the users in a hybrid VLC/WiFi system. Supplementing VLC by RF AP(s) is further shown to support mobility and decrease the handover overhead\cite{hand, wang2017optimization}. Authors in  \cite{wu2017access} show that the users that experience blockages with high rate should be served by the RF network. References \cite{JOCN, globecom} allocate the power and assign the users to VLC and RF APs to reduce the interference effect and to maximize the system capacity and fairness. References \cite{coop, hand, wang2017optimization, wu2017access, JOCN, globecom}, however, do not consider neither non-orthogonal multiple access schemes and do not allow any level of cooperation among users, unlike our current paper which addresses the rate-fairness balancing problem by means of adopting a Co-NOMA scheme in hybrid VLC/RF systems.

To further alleviate the VLC limitations, cooperation among APs is proposed in the recent literature, e.g., \cite{coop,hanzo_haas,TWC}, by merging the cells and boosting the users quality-of-service (QoS) by coordinating the APs transmissions. Such coordination allows the APs to cancel the interference \cite{chen2013joint,hanzo_haas, TWC}, increase the cell coverage \cite{coop}, and mitigate the blockages effect \cite{ chen2017performance}. Relaying is also investigated in VLC networks to extend the VLC coverage \cite{dual1,dual2, Kizilirmak2015}. In \cite{dual1} and \cite{dual2}, dual-hop hybrid VLC/RF links are proposed to serve uncovered users by means of relaying. Under such schemes, the relay harvests the energy and receives the signal through the first hop VLC link, and then forwards the signal to the receiver through the second hop RF link. The current paper adopts such a relaying scheme in a hybrid VLC/RF link, albeit under a different system model which considers a Co-NOMA scheme that gives the weak user the opportunity to be served either through the VLC link or through the hybrid VLC/RF link.

NOMA has extensively been investigated in RF networks, and has shown to enhance the spectral efficiency and the system fairness\cite{ding2017survey,liu2017nonorthogonal}. In the context of VLC networks, reference \cite{kizilirmak2015non} shows the superiority of NOMA over orthogonal-frequency division multiple access (OFDMA) with regard to sum-rate. The authors in \cite{ yin2016performance} evaluate the NOMA-VLC system and compare it to orthogonal multiple-access (OMA)-VLC scheme with and without QoS constraints. The authors in \cite{ yapici2018non} evaluate and compare the NOMA and OMA when the locations and the vertical orientations of the devices is changing.  For multiple APs, the work in \cite{ marshoud2016non } proposes a gain ratio power allocation (GRPA) method and compares it with the fixed power allocation method, when  the users' movement model follows the random walk model. For multi-cell VLC networks, the users in \cite{zhang2017user} are classified based on the received interference. In reference \cite{zhang2017user}, special resource blocks are assigned for the interfering users, while NOMA is applied for the interference-free users. Recently, the authors in \cite{Zhou2018} investigate the outage capacity for a limited system model, consisting of one VLC AP and two users (one is covered, while the other is out of the coverage). The covered user uses the Co-NOMA scheme to forward the uncovered user signal through the RF link, and the uncovered user does not have the option to be served directly through the VLC link.

In all the above studies, however, jointly allocating the power, selecting the link, and pairing the users have not been studied in the context of CO-NOMA in hybrid VLC/RF systems, which the current paper tackles. Additionally, to the authors' best knowledge, this is the first paper which proposes a scheme that provides the weak users with the possibility of being served through either the VLC link or through the hybrid VLC/RF link.


\subsection{Contributions}
Unlike the aforementioned papers, this paper introduces and evaluates a cooperation scheme among users in VLC networks based on hybrid VLC/RF Co-NOMA. More precisely, the paper considers a system model, consisting of one VLC AP and multiple users classified as either strong or weak users, strong and weak users. A transmission scheme based on hybrid VLC/RF Co-NOMA scheme is proposed, where the strong users receive their own signals through the direct VLC link. The paper transmission scheme then provides two options for each weak user, either to be served by the direct VLC link, or to be served by the relayed hybrid VLC/RF link that is provided by a strong user which can help the weak user by forwarding their signal through the RF link. This technique extends the VLC coverage area and helps serving the blocked users, which leads to improving the fairness and balancing the load in VLC systems. To this end, the paper addresses the problem of maximizing a network-wide weighted sum-rate, so as to jointly pair the strong and weak users, allocate their respective powers, and select the link of operation of every weak user (i.e., direct VLC or hybrid VLC/RF). The paper contributions can then be summarized as follows:
\begin{itemize}
\item The paper formulates a mixed discrete-continuous optimization problem that jointly pairs the users, selects the serving link for each weak user, and allocates the messages' power to maximize the weighted sum-rate of the system under user connectivity and maximum power constraints.

\item  The paper solves the mixed-integer non-convex optimization problem by proposing an iterative algorithm, which iteratively finds each of the optimization parameter by fixing all others. The paper particularly derives closed-form waterfilling-like solutions for the power allocation problem, and proposes well-chosen heuristics for finding the user pairing and link selection parameters. The paper also proposes updating the weights of the weighted-sum rate objective in an outer loop to achieve a proportional fairness of the system where, in each iteration, the weight of every user is determined to be inversely proportional to the long-term average rate of that user \cite{Yu2011}.



\item The paper compares the proposed solutions both to a simpler baseline approach, and to the conventional NOMA scheme. The paper simulation results illustrate how the proposed hybrid VLC/RF Co-NOMA scheme significantly improves the VLC network performance, both in terms of sum-rate and fairness, as compared to the conventional NOMA scheme.
\end{itemize}
The rest of this paper is organized as follows. The system and channel models are presented and discussed in Section II. Section III formulates the optimization problem and presents the proposed algorithms. Simulation results are illustrated in Section IV. Finally, we conclude the paper in Section V.
\section{System and Channel Models}
\subsection{System Model}

The system model considered in this paper consists of a VLC AP and $N_u=2K$ multiple users, where $K$ is the number of pairs. The users are classified into two sets, a set $\mathbb{U}_s$ of $K$ users, defined as strong users, which are the users with strong VLC channels, and a set  $\mathbb{U}_w$ of $K$ weak users, which are the users with the weak VLC channels. Fig. \ref{SM} shows an example of the considered network which serves three weak users, and three strong users. Without loss of generality, the paper assumes that the number of users are even, which facilitates the investigation of the proposed Co-NOMA technique. In the case where the number of users are odd, one can assume the existence of one more additional virtual user with a zero channel gain (and equivalently a zero achievable rate).

The paper adopts a Co-NOMA scheme, where the available bandwidth is divided equally into $K$ blocks. Each spectrum block is shared by a pair of users (one strong and one weak user). The weak user in each pair can be served either directly by the VLC AP through the VLC link, or by the dual-hop hybrid VLC/RF link through the paired strong user. In the case where the weak user is not well served by the VLC, the strong user would act as an energy harvesting relay which harvests the energy from the VLC AP (using the received visible light), and would then use it to forward the message to the paired weak user using the RF link.  Hereinafter, we illustrate the VLC  and RF channel model, the energy harvesting signals, and the transmission scheme.
\begin{figure}[!t]
\centering
\includegraphics[width=4in]{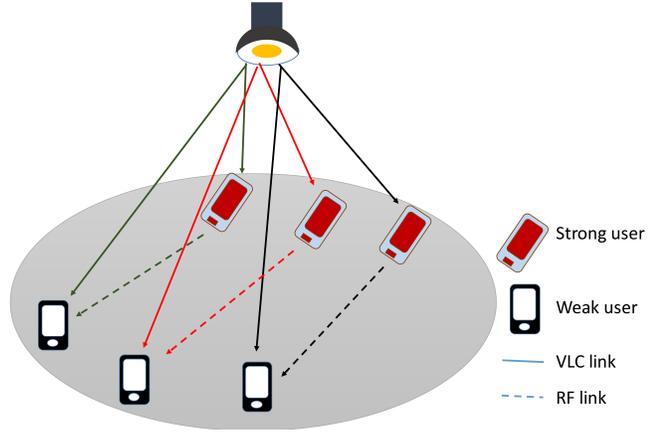}
\caption{An example of the considered system with three weak users, and three strong users.}
\label{SM}
\end{figure}
\subsection{VLC and RF Channel Model}
 According to \cite{wir}, the VLC channel between the AP and the $j$th user, denoted by $h_{j}$, is given by
\begin{equation}
\label{vlcch}
h_{j}=\frac{(m+1)A_{p}}{2\pi d_{j}^2} \cos^m(\phi)g_{of}\cos(\theta)f(\theta),
\end{equation}%
where $A_p$ is the photo-detector (PD) physical area, $m=-\left(\log_2(\cos(\theta_{\frac{1}{2}})\right)^{-1}$ is the Lambertian index, $\theta_{\frac{1}{2}}$ is the semi-angle of half power,  $d_{j}$ is the distance between the  AP and the $j$th user, $g_{of}$ is  the gain of the optical filter, $\phi$ is the LED  radiance angle, $\theta$ is the PD incidence angle, and $f(\theta)$  is the gain of the optical concentrator given by
\begin{equation}
\label{FoVE}
f(\theta)= \begin{cases}\frac{n^2}{\sin^2(\Theta)}, & \theta\leq\Theta; \\
0, & \theta>\Theta,
\end{cases}
\end{equation}
where $n$ is the refractive index and $\Theta$ is the semi-angle
of the user's field-of-view (FoV). We assume that if the LoS between the AP and the PD, the channel is zero, otherwise the channel is given by (\ref{vlcch}).

According to \cite{Perahia2013}, the RF channel gain in indoor environment between the $j$th user and the $i$th user is given by
\begin{equation}
G_{i,j}^{RF}=\vert H^{(RF)}_{i,j}\vert^2 10^{-\frac{L(d_{i,j})}{10}},
\end{equation}
where $H^{(RF)}_{i,j}$ is  the RF multipath propagation channel and $L(d_{i,j})$ is the path loss between the $j$th user and the $i$th user, where $d_{i,j}$ is the distance between $i$ and $j$ users.
\subsection{Transmission Scheme: Direct Links}
In the rest of this paper, we associate the subscript $i$ with weak users, and $j$ with strong users. Following the NOMA principle, suppose that one weak user $i$ is paired with one strong user $j$, the expression of the transmitted signal from the AP towards the user-pair $i$ and $j$ is given by
\begin{equation}
y_{i,j}=\nu\sqrt{P_{j}^{(s)}}s_j+\nu\sqrt{P_i^{(w)}}s_i+\nu b,
\end{equation}
where $P_j^{(s)}$ and $P_i^{(w)}$ are the powers of the strong and weak users assigned for $s_j$ and $s_i$ messages, respectively, $b$ is the direct-current (DC) that must be added to guarantee that the transmitted signal is non-negative, and $\nu$ is the proportionality factor of the electric to optical power conversion. The received signals at the $j$th and  $i$th users, respectively, are given by
\begin{equation}
\label{y1}
y_j=\nu\rho h_j^{(s)}(\sqrt{P_{j}^{(s)}}s_j+\sqrt{P_i^{(w)}}s_i)+\nu\rho h_j^{(s)}b+n_j,
\end{equation}
\begin{equation}
\label{y2}
y_i=\nu\rho h_i^{(w)}(\sqrt{P_{j}^{(s)}}s_j+\sqrt{P_i^{(w)}}s_i)+\nu \rho h_i^{(w)}b+n_i,
\end{equation}
where  $h_j^{(s)}$ and $ h_i^{(w)}$ are the channels of the $j$th strong user and the $i$th weak user, respectively, $\rho$ is the detector responsivity, and $n_j$ or $n_i$ are the noise components which can be modeled as real zero-mean additive white Gaussian noise variable (AWGN) with variance $\sigma^2 = N_vB_v$, where $N_v$ is the noise power spectral density (PSD) and $B_v$ is the modulation bandwidth. The received DC part $\rho  \nu h_j^{(s)} b$ at the strong user can be separated by a capacitor and directed to the energy harvesting circuit \cite{wang2015design}, while at the weak user, the DC part $\rho \nu h_i^{(w)}b$ can be removed using a capacitor. The strong user detects the weak user's message $s_i$ and then removes it or directs it to the weak user, while the weak user decodes his own signal $s_i$ by treating the strong user's signal $s_j$ as interference.  Hence, the achievable data rate of the strong  user signal can be approximated by  \cite{yin2016performance, zhang2017user, yang2017fair}
\begin{equation}
R_j^{(s)}(P_j^{(s)})=\frac{B_v}{2K}\log_2\left( 1+\frac{\nu^2\rho^2 h_j^{(s)2}P_j^{(s)}}{B_v N_v/K}\right).
\end{equation}
The achievable data rate of the weak user data decoded by the strong user is given by
\begin{equation}
\label{VLCws}
R_{j\rightarrow i}^{(w)}(P_i^{(w)},P_j^{(s)})=\frac{B_v}{2K}\log_2\left( 1+\frac{\nu^2\rho^2 h_j^{(s)2}P_i^{(w)}}{B_vN_v/K+\nu^2\rho^2 h_j^{(s)2}P_j^{(s)}}\right).
\end{equation}
The achievable data rate received at the weak user from the VLC direct link (DL) is given by
\begin{equation}
R_{i,DL}^{(w)}(P_i^{(w)},P_j^{(s)})=\frac{B_v}{2K}\log_2\left( 1+\frac{\nu^2\rho^2 h_i^{(w)2}P_i^{(w)}}{B_vN_v/K+\nu^2\rho^2 h_i^{(w)2}P_j^{(s)}}\right).
\end{equation}
\subsection{Transmission Scheme: Relayed Links}
The paper assumes that the strong user can work also as a relay that has the ability to harvest the energy from the light intensity, and then to utilize it to forward the decoded weak user's signal.
To harvest the energy, a capacitor separates the DC component from the received electrical signal and forwards it to the energy harvesting circuit \cite{wang2015design, obeed2018dc, simul2}. The received energy at the $j$th user can be expressed as \cite{solar}
\begin{equation}
\label{EH}
E_j=f \rho \nu V_t h_jb\ln(1+\frac{\rho h_j\nu b}{I_0}).
\end{equation}
where $V_t$ is the thermal voltage, $f$ is the fill factor,  and $I_0$ is the dark saturation current of the PD.
Suppose that the amplitude of the  transmitted signal  is $A$. The DC-bias and the signal amplitude $A+b$ must, therefore, be within the maximum and minimum input currents to make sure that the LED transmitter is operating in the linear region.  In other words, let $I_H$ and $I_L$ be the maximum and minimum limits of the input currents for the LEDs that guarantee a linear output optical power. The constraint $A\leq \min(b-I_L,I_H-b)$ must then be achieved. The deriving power to the LED is related to $A$ and $b$ by $P_{max}=(A-b)^2$. This means that the maximum allowed deriving power $P_{max}$ at the AP is a decreasing function of the DC-bias $b$. The DC-bias can be optimized to balance between the received harvested energy and the transmit power, but this is out of the scope of this paper. Hence, the DC-bias at the AP is assumed to be fixed and is given by $b=\frac{I_H+I_L}{2}$ which maximizes the total transmit power \cite{obeed2018dc}. Therefore, the maximum allowed deriving power is given by
\begin{equation}
P_{max}=(\frac{I_H-I_L}{2})^2.
\end{equation}
The strong user is assumed to be able to receive the light signal and transmit the RF signal at the same time. However, the energy storage device cannot charge and discharge at the same time (i.e., the receiver cannot harvest the energy and transmit the data at the same time). Hence, suppose that $T_1$  is the time spent to charge the battery and $T_2$ is the time used to discharge or re-transmit data through the RF link. Therefore, the RF re-transmission power is given by $P_{j,RF}=\frac{E_jT_1}{T_2}$. Under the assumption that $T_1=T_2$, the achievable data rate of the weak user that can be offered by the strong user $j$ through the RF link is given by
\begin{equation}
\label{RFj}
R_{i,j}^{RF}=\frac{B_{f,j}}{2}\log\left(1+\frac{G_{i,j}^{RF}P_{j,RF}}{B_{f,j}N_{RF}}\right),
\end{equation}
where $B_{f,j}$ is the RF modulation bandwidth at the user $j$ and $N_{RF}$ is the PSD of the RF signal. If the number of weak users that are served through relayed VLC/RF links is more than one, the RF bandwidth must be divided between weak users; otherwise, additional interference terms would arise. This paper assigns orthogonal RF bandwidth for each pair connected through the RF link equally,  so as to nullify the wireless interference among the active RF links. Suppose that the number of weak users that are served through RF links is $N_f$ and the total modulation bandwidth available for RF transmission is $B_f$, the modulation bandwidth at user $j$ is $B_{f,j}=\frac{B_f}{N_f}$.

From (\ref{RFj}) and (\ref{VLCws}), the achievable data rate at the weak user through the hybrid relayed link (RL) can be expressed as
\begin{equation}
R_{i,RL}^{(w)}(P_i^{(w)},P_j^{(s)})=\min\left(R_{i,j}^{RF},R_{j\rightarrow i}^{(w)}(P_i^{(w)},P_j^{(s)})\right).
\end{equation}

\section{Problem Formulation and Solutions}
\begin{table*}[!t]
\centering
\caption{Data rates notations}
\label{tableR}
\begin{tabular}{|p{.18\textwidth} | p{.36\textwidth} | p{.37\textwidth} |}
\hline
  \textbf{Symbol} &   \textbf{Definition} &   \textbf{Expression}\\
 \hline

  $R_j^{(s)}(P_j^{(s)})$ & The achievable data rate of the $j$th strong  user signal  &  $R_j^{(s)}(P_j^{(s)})=\frac{B_v}{2K}\log_2\left( 1+\frac{\nu^2\rho^2 h_j^{(s)2}P_j^{(s)}}{B_v N_v/K}\right)$ \\
  \hline
 $R_{j\rightarrow i}^{(w)}(P_i^{(w)},P_j^{(s)})$ & The achievable data rate of the weak user data decoded by the strong user & $R_{j\rightarrow i}^{(w)}(P_i^{(w)},P_j^{(s)})=\frac{B_v}{2K}\log_2\left( 1+\frac{\nu^2\rho^2 h_j^{(s)2}P_i^{(w)}}{B_vN_v/K+\nu^2\rho^2 h_j^{(s)2}P_j^{(s)}}\right)$\\
 \hline

$R_{i,DL}^{(w)}(P_i^{(w)},P_j^{(s)})$ & The achievable data rate received at the weak user from the VLC direct link (DL) & $R_{i,DL}^{(w)}(P_i^{(w)},P_j^{(s)})=\frac{B_v}{2K}\log_2\left( 1+\frac{\nu^2\rho^2 h_i^{(w)2}P_i^{(w)}}{B_vN_v/K+\nu^2\rho^2 h_i^{(w)2}P_j^{(s)}}\right),$\\
\hline
$R_{i,j}^{RF}(\mathbf{x})$ & The achievable data rate of the $i$th weak user offered by the $j$th strong user through the RF link   &$R_{i,j}^{RF}(\mathbf{x})=\frac{B_{f}}{2\sum_{x_i}^Kx_i}\log\left(1+\frac{ G_{i,j}^{RF}P_{j,RF}}{B_fN_{RF}/ \sum_{i}^Kx_i}\right).$ \\
\hline

$R_{i,RL}^{(w)}(P_i^{(w)},P_j^{(s)},\mathbf{x})$& The achievable data rate at the weak user through the hybrid relayed link (RL) &
$R_{i,RL}^{(w)}(P_i^{(w)},P_j^{(s)})=\min\left(R_{i,j}^{RF}(\mathbf{x}),R_{j\rightarrow i}^{(w)}(P_i^{(w)},P_j^{(s)})\right)$\\
\hline

$R_{i,j}(P_i^{(w)},P_j^{(s)},z_{i,j},\mathbf{x})$& The summation of the achievable data rate of the $i$th weak user and $j$th strong user & $R_{i,j}(P_i^{(w)},P_j^{(s)},z_{i,j},\mathbf{x})= z_{i,j}R_j^{(s)}(P_j^{(s)})
 +z_{i,j}(1-x_i)R_{i,DL}^{(w)}(P_i^{(w)},P_j^{(s)})
+z_{i,j}x_{i}R_{i,RL}^{(w)}(P_i^{(w)},P_j^{(s)},\mathbf{x})$\\

\hline
\end{tabular}
 \end{table*}

Our goal in this paper is to maximize the weighted sum of the achievable data rates under user connectivity and transmit power constraints. To formulate and tackle the problem, we should answer the three interlinked questions: 1) How should the users be paired? 2) What is the serving link of each weak user (i.e., direct VLC link or relayed VLC/RF link)? 3) How should the total power at the AP be allocated among users? This section formulates these questions as an optimization problem to jointly obtain the power allocation vector $\mathbf{P}$, the pairing index matrix $\mathbf{Z}$, and the link selection index vector $\mathbf{x}$ to maximize the network-wide weighted sum-rate. {{To simplify the paper presentation and avoid the confusion of the notations of the data rates, we introduce Table \ref{tableR} which defines the different achievable data rates, their symbols, and expressions. }}

Define the user pairing $K\times K$ matrix $\mathbf{Z}$, where the entries of $\mathbf{Z}$ are defined as follows:
\begin{equation}
z_{i,j}= \begin{cases} 1, & $if weak user $i$ is paired with strong user $j$,$ \\
0, & $otherwise$.
\end{cases}
\end{equation}
Define the link selection indicator vector $\mathbf{x}$, where the entries of $\mathbf{x}$ are defined as follows:
 \begin{equation}
 \label{xind}
x_{i}= \begin{cases} 1, & $if user $i$ is served through the hybrid RF/VLC link,$ \\
0, & $if user $i$ is served through the direct VLC link.$
\end{cases}
\end{equation}
From (\ref{xind}), the number of weak users that are connected through RF link is given by $N_f=\sum_{i=1}^Kx_i$. The rate in (\ref{RFj}) becomes, therefore, a function of $\mathbf{x}$ as shown in Table \ref{tableR}.
The summation of the achievable data rate of the $i$th weak user and the $j$th strong user is given by
\begin{multline}
\label{Rij}
\resizebox{0.91\hsize}{!}{$R_{i,j}(P_i^{(w)},P_j^{(s)},z_{i,j},\mathbf{x})=
 z_{i,j}R_j^{(s)}(P_j^{(s)})
 +z_{i,j}(1-x_i)R_{i,DL}^{(w)}(P_i^{(w)},P_j^{(s)})
+z_{i,j}x_{i}R_{i,RL}^{(w)}(P_i^{(w)},P_j^{(s)},\mathbf{x})$}.
\end{multline}
The weighted sum-rate of the system users is given by
\begin{multline}
\label{Rij2}
\sum_{i=1}^K\sum_{j=1}^K R_{i,j}(P_i^{(w)},P_j^{(s)},z_{i,j},x_{i},w_i^{(w)},w_j^{(s)})=
\sum_{i=1}^K\sum_{j=1}^K w_j^{(s)} z_{i,j}R_j^{(s)}(P_j^{(s)})\\
+w_iz_{i,j}(1-x_i)R_{i,DL}^{(w)}(P_i^{(w)},P_j^{(s)})
+w_i^{(w)}z_{i,j}x_{i}\min\left(R_{i,j}^{RF}(\mathbf{x}),R_{j\rightarrow i}^{(w)}(P_i^{(w)},P_j^{(s)})\right),
\end{multline}
where the weights $w_i^{(w)}$ and $w_j^{(s)}$ ($\forall \ i,j$) are imposed to balance between the system sum-rate and the system fairness. 
Based on expression (\ref{Rij2}), the considered optimization problem can now be expressed mathematically as follows:
\begin{subequations}
\label{EHM}
\begin{eqnarray}
&\displaystyle\max_{\mathbf{P}, \mathbf{z}, \mathbf{x}}&  \sum_{i=1}^K\sum_{j=1}^K R_{i,j}(P_i^{(w)},P_j^{(s)},z_{i,j},x_{i,j},w_i^{(w)},w_j^{(s)})\\
\label{EHMb}
&s.t.&   \sum_{i=1}^K\sum_{j=1}^Kz_{i,j}(P_i^{(w)}+P_j^{(s)})\leq P_{max}\\
\label{EHMc}
&&   \sum_{j=1}^K z_{i,j}=1,\ \forall i,\ 
 \sum_{i=1}^K z_{i,j}=1,\ \forall j,\ z_{i,j}\in \lbrace 0,1 \rbrace,\ \forall i,j\\
\label{EHMf}
&& x_{i}\in \lbrace 0,1  \rbrace, \ \forall i\\
\label{EHMg}
&& 0\leq P_j^{(s)}\leq P_i^{(w)} \ \forall i,j, 
\end{eqnarray}
\end{subequations}
\noindent where the optimization is over the continuous power variable $\mathbf{P}$, and the discrete association variables $\mathbf{z}$ and $\mathbf{x}$. Constraint (\ref{EHMb}) is imposed for the maximum transmit power.
Constraints in (\ref{EHMc}) guarantee that each strong user is paired only with one weak user. Constraint (\ref{EHMf}) imposes that the weak user receives the information either from the direct VLC link or from the relayed hybrid VLC/RF link. Constraint (\ref{EHMg}) is the power constraint that is imposed for successive interference cancellation in NOMA system.

The problem in (\ref{EHM}) is a challenging mixed non-convex combinatorial optimization problem.  
The paper tackles such a difficult problem through a heuristic approach. The main idea is to solve the problem for each parameter when the other parameters are fixed. Specifically, for the power allocation problem, we find closed-form solutions for the formulated non-convex optimization problem. For the user pairing problem, we use the Hungarian method for strong-to-weak users assignment problem. For the link selection problem, we find the optimal solution by first generating a $K\times K$ matrix that reduces the number of candidate vectors from $2^K$ to $K$. The overall algorithm then iterates among the above three steps so as to jointly determine the user pairing, link selection, and power allocation. Although such solution does not lead to the optimal solution of problem (\ref{EHM}), the simulations section of the paper illustrates how the proposed solution notably improves the performance of VLC network in terms of sum-rate and fairness.

\subsection{Power Allocation}
\label{PAWS}

This section solves the optimization problem (\ref{EHM}) when the $\mathbf{Z}$ matrix and $\mathbf{x}$ vector are fixed. In particular, we find the power allocation for fixed user pairing and link selection. Under the given $\mathbf{Z}$ and $\mathbf{x}$, the weighted sum maximization problem can be formulated as follows
\begin{subequations}
\label{EHM2}
\begin{eqnarray}
&\displaystyle\max_{\mathbf{P}}&  \sum_{i=1}^K\sum_{j=1}^K R_{i,j}(P_i^{(w)},P_j^{(s)},z_{i,j},x_{i,j},w_i^{(w)},w_j^{(s)})\\
\label{EHM2b}
&s.t.&   \sum_{i=1}^K\sum_{j=1}^Kz_{i,j}(P_i^{(w)}+P_j^{(s)})\leq P_{max}\\
\label{EHM2c}
&& 0\leq P_j^{(w)}\leq P_i^{(s)} \ \forall i,j.
\end{eqnarray}
\end{subequations}

Clearly, problem (\ref{EHM2}) is non-convex, since the objective function is not concave.  However, in the following, we provide closed-form solutions for the power allocation. Define a variable $q_{i,j}$ as the power budget of a pair consisting of the $i$th weak user and $j$th strong user, i.e., $q_{i,j}=z_{i,j}(P_i^{(w)}+P_j^{(s)})$. Then, we can solve problem (\ref{EHM2}) by first breaking the problem into into $K$ sub-problems to find $P_i^{(w)}$
and $P_j^{(s)},\ (\forall i, j),$ the solution of which depends on one
main problem which solves the power budgets $q_{i,j}, \forall  i, j$.
Hence, for the $k$th problem, suppose that the weak user $i$ is paired with the strong user $j$ (i.e., $z_{i,j}=1$). We then distinguish between two cases, i.e., either the weak user is served by the hybrid VLC/RF link (called Case 1 in the rest of the paper), or the weak user is served by the VLC link (called Case 2 in the rest of the paper).

\subsubsection{Case 1} The users $i$ and $j$ are paired, and the weak user $i$ is served through the hybrid VLC/RF link. For a given $q_{i,j}$, the problem of finding $P_i^{(w)}$
and $P_j^{(s)}$ in this case can be formulated as follows
\begin{subequations}
\label{EHM3}
\begin{eqnarray}
&\displaystyle\max_{P_i^{(w)},P_j^{(s)}}&
 w_j^{(s)}R_j^{(s)}(P_j^{(s)})+w_i^{(w)} \min(R_{i,j}^{RF}(\mathbf{x}),R_{j\rightarrow i}^{(w)}(P_i^{(w)},P_j^{(s)}))\\
\label{EHM3b}
&s.t.&   P_i^{(w)}+P_j^{(s)}= q_{i,j}\\
\label{EHM3c}
&& 0\leq P_j^{(s)}\leq P_i^{(w)}.
\end{eqnarray}
\end{subequations}

To solve problem (\ref{EHM3}), note first that the problem above must be carefully treated because of the $\min$ term in the objective function. Also, for a fixed $\mathbf{x}$, observe that the function $R_{i,j}^{RF}(\mathbf{x})$ is fixed, as it is not a function of the power variables $P_i^{(w)}$ and $P_j^{(s)}$.

\begin{remark}
\label{rem1}
\emph{In problem (\ref{EHM3}), the resulting optimal value of $R_{j\rightarrow i}^{(w)}(P_i^{(w)},P_j^{(s)})$ must be less than or equal to the resulting optimal value of $R_{i,j}^{RF}(\mathbf{x})$. To prove that, we can see that increasing $P_i^{(w)}$ (decreasing $P_j^{(s)}$) increases $R_{j\rightarrow i}^{(w)}(P_i^{(w)},P_j^{(s)})$ and decreases the achievable data rate of the strong user at the same time. This means that increasing $P_i^{(w)}$ to have $R_{j\rightarrow i}^{(w)}(P_i^{(w)},P_j^{(s)})$ larger than $R_{i,j}^{RF}(\mathbf{x})$ would fix the data rate of the weak user at $R_{i,j}^{RF}(\mathbf{x})$ and would decrease the data rate of the strong user. In other words, if $\hat{P}_i$ is the value that would make $R_{i,j}^{RF}(\mathbf{x})=R_{j\rightarrow i}^{(w)}(\hat{P}_i,P_j^{(s)})$, the optimal $P_i^{(w)}$ must be less than or equal to $\hat{P}_i$.}
\end{remark}


To solve problem (\ref{EHM3}), we first set that
$ \min(R_{i,j}^{RF}(\mathbf{x}),R_{j\rightarrow i}^{(w)}(P_i^{(w)},P_j^{(s)}))=R_{j\rightarrow i}^{(w)}(P_i^{(w)},P_j^{(s)})$  and solve the problem. If the resulting $R_{j\rightarrow i}^{(w)}(P_i^{(w)},P_j^{(s)})$ satisfies the inequality $R_{j\rightarrow i}^{(w)}(P_i^{(w)},P_j^{(s)}) \leq R_{i,j}^{RF}(\mathbf{x})$, then the resulting powers are the optimal solution. Otherwise (i.e., if $R_{j\rightarrow i}^{(w)}(P_i^{(w)},P_j^{(s)}) > R_{i,j}^{RF}(\mathbf{x})$), the optimal solution is the value of the power $P_i^{(w)}$ that achieves that $R_{j\rightarrow i}^{(w)}(P_i^{(w)},P_j^{(s)}) = R_{i,j}^{RF}(\mathbf{x})$  and $P_j^{(s)}=q_{i,j}-P_i^{(w)}$.


\begin{proposition}
\label{prop1}
In problem (\ref{EHM3}), if we replace the term $ \min(R_{i,j}^{RF}(\mathbf{x}),R_{j\rightarrow i}^{(w)}(P_i^{(w)},P_j^{(s)}))$ by $R_{j\rightarrow i}^{(w)}(P_i^{(w)},P_j^{(s)}))$, the optimal power allocations are given by $P_j^{(s)}=\eta_{i,j,1}$,  where
\begin{equation}
\label{eta_1}
\eta_{i,j,1}=\frac{-1+\sqrt{1+q_{i,j}\Psi_j^{(s)} }}{\Psi_j^{(s)}},
\end{equation}
and
\begin{equation}
\label{P_i}
P_i^{(w)}=q_{i,j}-\eta_{i,j,1},
\end{equation}
where $\Psi_j^{(s)}=\frac{\rho^2 h_j^{(s)2}}{B_vN_v/K}$ and $\Psi_i^{(w)}=\frac{\rho^2 h_i^{(w)2}}{B_vN_v/K}$.
\end{proposition}
\begin{proof} See Appendix \ref{AppA}.
\end{proof}

 Now, if equation (\ref{eta_1}) and (\ref{P_i}) achieve that $ \min(R_{i,j}^{RF}(\mathbf{x}),R_{j\rightarrow i}^{(w)}(P_i^{(w)},P_j^{(s)}))=R_{j\rightarrow i}^{(w)}(P_i^{(w)},P_j^{(s)}))$, Proposition \ref{prop1} provides the optimal solution for (\ref{EHM3}); otherwise, the values of $P_j^{(s)}$ and $P_i^{(w)}$ must be modified to comply with Remark \ref{rem1} to satisfy  $R_{j\rightarrow i}^{(w)}(P_i^{(w)},P_j^{(s)}) = R_{i,j}^{RF}(\mathbf{x})$. By solving this equation ($R_{j\rightarrow i}^{(w)}(P_i^{(w)},P_j^{(s)}) = R_{i,j}^{RF}(\mathbf{x})$), we  obtain $P_j^{(s)}=\eta_{i,j,2}$, where $\eta_{i,j,2}$ is given by
\begin{equation}
\label{eta2}
\eta_{i,j,2}=\frac{q_{i,j}\Psi_j^{(s)}+1-A}{A\Psi_j^{(s)}},
\end{equation}
where $A=2^{2R_{i,j}^{RF}(\mathbf{x})K/B_v}$ and $P_i^{(w)}$ can be given by $P_i^{(w)}=q_{i,j}-\eta_{i,j,2}$. We should note that the values of both $w_i^{(w)}$ and $w_j^{(s)}$ must be selected carefully because, as we  show later, the budget $q_{i,j}$ is a function of both weights $w_i^{(w)}$ and $w_j^{(s)}$.

\subsubsection{Case 2} The users $i$ and $j$ are paired, and the weak user $i$ is served through the direct VLC link. The problem in this case can be formulated as

\begin{subequations}
\label{EHM4}
\begin{eqnarray}
&\displaystyle\max_{P_i^{(w)},P_j^{(s)}}&
 w_j^{(s)}R_j^{(s)}(P_j^{(s)})+w_i^{(w)} R_{i,DL}^{(w)}(P_i^{(w)},P_j^{(s)})\\
\label{EHM4b}
&s.t.&   P_i^{(w)}+P_j^{(s)}= q_{i,j}\\
\label{EHM4c}
&& 0\leq P_j^{(s)}\leq P_i^{(w)}.
\end{eqnarray}
\end{subequations}
Because of the interference term in $R_{i,DL}^{(w)}(P_i^{(w)},P_j^{(s)})$, the optimization problem (\ref{EHM4}) is still nonconvex. However,  the authors of \cite{Zhu2017} tackle a  problem of similar structure and show that the optimal solution has a closed-form.
By setting the derivative of the objective function of (\ref{EHM4}) equal to zero, we obtain
\begin{multline}
\label{D1}
\resizebox{0.92\hsize}{!}{$\frac{d}{dP_j^{(s)}}\big[w_j^{(s)}R_j^{(s)}(P_j^{(s)})+w_i^{(w)} R_{i,DL}^{(w)}(P_i^{(w)},P_j^{(s)})\big]=
\frac{w_j^{(s)}B_v}{2K(1/\Psi_j^{(s)}+P_j^{(s)})}-\frac{w_i^{(w)}B_v}{2K(1/\Psi_i^{(w)}+P_j^{(s)})}=0$}
\end{multline}
 One can readily verify that the second derivative of the objective function is always negative if $\Psi_j^{(s)}\geq \Psi_i^{(w)}$ and $w_i^{(w)}/w_j^{(s)}<\Psi_j^{(s)}/\Psi_i^{(w)}$, i.e., the objective function is concave in such case. Therefore, the optimal solution can be obtained by solving equation (\ref{D1}), which leads to a unique root $P_j^{(s)}=\Omega_{i,j}$, where $\Omega_{i,j}$ is given by
\begin{equation}
\label{omga}
\Omega_{i,j}=\frac{w_i\Psi_i^{(w)}-w_j^{(s)}\Psi_j^{(s)}}{\Psi_j^{(s)}\Psi_i^{(w)}(w_j^{(s)}-w_i^{(w)})},
\end{equation}
where the conditions that $w_i^{(w)}/w_j^{(s)}<\Psi_j^{(s)}/\Psi_i^{(w)}$ and $q_{i,j}>2\Omega_{i,j}$ must be satisfied, and $P_i^{(w)}=q_{i,j}-\Omega_{i,j}$.

\subsubsection{Determination of $q_{i,j}$}
The previous analysis of both Case 1 and Case 2 allows to determine the powers $P_j^{(s)}$ and $P_i^{(w)}$ as a function of $q_{i,j}$, i.e., equations (\ref{eta_1}), (\ref{eta2}), and (\ref{omga}). By substituting the corresponding expressions of $P_j^{(s)}$ and $P_i^{(w)}$ in problem (\ref{EHM2}), one can formulate the following optimization problem:
\begin{subequations}
\label{EHM5}
\begin{eqnarray}
\nonumber
&\max_{q_{i,j}}&  \sum_{i=1}^K\sum_{j=1}^K x_iz_{i,j}\big(w_{j}^{(s)}F_j^{(s)}(q_{i,j})\\
\nonumber
&& +w_{i}^{(w)}F_{i}^{(w)}(q_{i,j})\big)+\sum_{i=1}^K\sum_{j=1}^K (1-x_i)z_{i,j}w_i^{(w)} \frac{B_v}{2K}\log_2(1+\Omega_{i,j}\Psi_j^{(s)})\\
&&+\sum_{i=1}^K\sum_{j=1}^K (1-x_i)z_{i,j}w_i^{(w)} \frac{B_v}{2K}\log_2\left(\frac{q_{i,j}\Psi_i^{(w)}+1}{\Omega_{i,j}\Psi_i^{(w)}+1}\right)\\
\label{EHM5b}
&s.t.&   \sum_{i=1}^K\sum_{j=1}^K q_{i,j}= P_{max},\\
&& q_{i,j}\geq 0, \ \forall i,j,
\end{eqnarray}
\end{subequations}
where  $F_j^{(s)}(q_{i,j})$ and $F_i^{(w)}(q_{i,j})$ can be written either as
\begin{equation}
 \label{Fj1}
 F_j^{(s)}(q_{i,j})=R_j^{(s)}(\eta_{i,j,1})=\frac{B_v}{2K}\log_2(\sqrt{\Psi_j^{(s)}q_{i,j}+1})
 \end{equation} and
 \begin{equation}\label{Fi1}
 F_i^{(w)}(q_{i,j})=R_{i,RL}^{(w)}(\eta_{i,j,1})=\frac{B_v}{2K}\log_2(\sqrt{\Psi_j^{(s)}q_{i,j}+1})
 \end{equation}
  or as
 \begin{equation}\label{Fj2}
 F_j^{(s)}(q_{i,j})=R_j^{(s)}(\eta_{i,j,2})=\frac{B_v}{2K}\log_2(\frac{\Psi_j^{(s)}q_{i,j}+1}{A})\end{equation} and
 \begin{equation}\label{Fi2}
 F_i^{(w)}(q_{i,j})=R_{i,RL}^{(w)}(\eta_{i,j,2})=R_{i,j}^{RF}(\mathbf{x}).
 \end{equation}
 At this stage, we cannot decide what are the exact expressions of the functions $F_j^{(s)}(q_{i,j})$ and $F_i^{(w)}(q_{i,j}),\ \forall i,j$, because both functions depend on  whether the optimal power allocation of the pair is given by (\ref{eta_1}) or by (\ref{eta2}). Since $q_{i,j}$ is still unknown, we next derive a closed-form waterfilling-like solution for $q_{i,j}$ (if $x_i=1$), for both possible expressions of  $F_j^{(s)}(q_{i,j})$ and $F_i^{(w)}(q_{i,j})$.
\begin{proposition}\label{prop2}
The optimal solution of problem (\ref{EHM5}) is expressed as
\begin{equation}\label{qij1}
q_{i,j}=\left[\frac{wB_v}{2K\lambda}-\frac{1}{\Psi_j^{(s)}}\right]^+,
\end{equation}
where $w=w_i^{(w)}$ if $x_i=0$, $w=w_j^{(s)}$ if $x_i=1$, $[n]^+$ means $\max(0,n)$, and $\lambda$ is the dual variable related to the total transmitting power constraint (\ref{EHM5b}) and can be found by substituting (\ref{qij1}) in constraint (\ref{EHM5b}).
\end{proposition}
\begin{proof} See Appendix \ref{AppB}.
\end{proof}


The steps of the power allocation algorithm are presented in Algorithm 1 below, which explains how to best allocate the powers of the users (for fixed users’ pairing and link selection).
\begin{algorithm}
 \caption{Allocating the power for users (under given users' pairing and link selection) }
 \label{G1}
 \begin{enumerate}
 \item Find $\lambda$ by substituting (\ref{qij1}) in constraint (\ref{EHM5b}).
 \item For each pair (when $z_{i,j}=1$), if the corresponding $x_i=0$, find $q_{i,j}$ using (\ref{qij1}) ($w=w_i^{(w)}$) and $P_j^{(s)}=\Omega_{i,j}$.
\item If the corresponding $x_i=1$, find $q_{i,j}$ using (\ref{qij1}) ($w=w_j^{(w)}$) and find $F_i^{(w)}(q_{i,j})$ using (\ref{Fi2}). If $F_i^{(w)}(q_{i,j})<=R_{i,j}^{(x)}$, the optimal $P_j^{(s)}$ is given by $\eta_{i,j,1}$; otherwise, the optimal $P_j^{(s)}$ is given by $\eta_{i,j,2}$, and the optimal $P_i^{(w)}=q_{i,j}-P_j^{(s)}$.
 \end{enumerate}
 \end{algorithm}
%

\subsection{Updating the Weights} To best balance the users rates across the network, this paper chooses to update the weights using the classical proportional fairness approach, e.g., see \cite{Yu2011} and references therein. More specifically,  we update the weights in an outer loop by setting  $w_i^{(w)}=\frac{1}{\bar{R_i}^{(w)}}$ and $w_j^{(s)}=\frac{1}{\bar{R_j}^{(s)}}$, where $\bar{R_i}^{(w)}$ and $\bar{R_j}^{(s)}$  are the long term average rates of the $i$th weak user and $j$th strong user, respectively \cite{Yu2011}. Consider the paired users $i$ and $j$. If $x_i=0$, the weights are updated in such a way that satisfies the condition:  $w_i^{(w)}/w_j^{(s)}<\Psi_j^{(s)}/\Psi_i^{(w)}$, i.e., to guarantee a positive power value for the strong user. Hence, through the outer loop iterations, if $\bar{R}_j$ is dropped below $\bar{R}_i$, we select $w_j^{(s)}=\alpha w_i^{(w)}$, where $\alpha\lessapprox 1$, i.e., $\alpha$ is strictly less than 1 (yet sufficiently close to 1).
\subsection{User Pairing Optimization}
 Section \ref{PAWS} determines the power allocation for given user pairing and link selection. In this section, we provide how to pair the users under a given fixed power allocation and link selection. For the given fixed power allocation and fixed link selection, the optimization problem (\ref{EHM}) can be formulated as 
\begin{subequations}
\label{EHM6}
\begin{eqnarray}
&\displaystyle\max_{\mathbf{Z}}&   \sum_{i=1}^K\sum_{j=1}^K w_j^{(s)} z_{i,j}R_j^{(s)}(P_j^{(s)})+w_i^{(w)}z_{i,j}(1-x_i)R_{i,DL}^{(w)}(P_i^{(w)},P_j^{(s)})\\
\nonumber
&&+w_i^{(w)}z_{i,j}x_{i}\min\left(R_{i,j}^{RF}(\mathbf{x}),R_{j\rightarrow i}^{(w)}(P_i^{(w)},P_j^{(s)})\right)\\
\label{EHM6b}
&s.t.&   \sum_{j=1}^K z_{i,j}=1,\ \forall i,\\
\label{EHM6d}
&& \sum_{i=1}^K z_{i,j}=1,\ \forall j,\\
\label{EHM6e}
&& z_{i,j}\in \lbrace 0,1 \rbrace,\ \forall i,j
\end{eqnarray}
\end{subequations}
The above optimization problem is a integer linear programming, which simply pairs each of the weak users to one (and only one) strong user based on the utility values. Such problem is considered as a one-to-one linear assignment problem \cite{Kuhn2005}, where it can be solved using one of the conventional matching algorithms, e.g., the Hungarian method \cite{Kuhn2005}.
\subsection{Link Selection Optimization} This section now focuses on the link selection problem for fixed power allocation and user pairing, which are determined in the previous subsections. The problem of selecting the optimal links for the weak users, i.e., either to be served by the VLC AP, or via the strong user which relays the information from the VLC AP to the weak user through the hybrid VLC/RF links (even for  given power allocation and user pairing) is not easy to be tackled. This paper addresses the problem by first generating a $K\times K$ matrix that reduces the number of candidates vectors $\mathbf{x}$  from $2^K$ to $K$ vectors (for a fixed user pairing and power allocation). In particular, define a matrix $S$, where the first row in $S$ hosts  the rates of the weak users coming from the relayed links subtracted form the rates coming from the direct link when $\sum_{i=1}^Kx_i=1$. Similarly, the second row in $S$ hosts  the rates of the weak users coming from the relayed links subtracted form the rates coming from the direct link when $\sum_{i=1}^Kx_i=2$. Construct all other rows of the matrix $S$ in a similar fashion, i.e., the $K$th row hosts  the rates of the weak users coming from the relayed links subtracted form the rates coming from the direct link, where $\sum_{i=1}^Kx_i=K$.

Select afterwards the vector $\mathbf{x}$ corresponding to the highest values of each row. More specifically, consider row $k$ of matrix $S$. The highest $k$ values of the row are then set to 1; the other entries are set to 0.   This results in having $K$ different $\mathbf{x}$ vectors, out of which we select the one that maximizes the objective weighted sum function. Such method has a polynomial computational complexity ($O(K\times K+K)$), which is significantly simpler than the complexity of the exhaustive search, i.e., $O(2^K)$. In addition, the proposed approach guarantees the optimal solution because it finds the optimal $\mathbf{x}$ vector for each case of $\sum_{i=1}^Kx_i$ and tries all of them to select the maximizing one.

\subsection{Overall Algorithm}

Now that each of the continuous and discrete variables of problem (\ref{EHM}) are determined separately, the paper proposes solving (\ref{EHM}) using an iterative algorithm that finds the three variables in an alternative way. In particular, we first initialize the user pairing and link selection, and find the powers using the proposed closed-form solutions based on the given user pairing and link selection initial values. Afterwards, we find the $\mathbf{Z}$ matrix based on the found allocated power and link selection. The link selection vector is finally updated to maximize the weighted sum objective function. These steps are repeated until convergence. The steps of the overall solution are given in Algorithm \ref{Algor2} Table.

\begin{algorithm}
 \caption{Overall algorithm for joint power allocation, user pairing, and link selection}
 \label{Algor2}
 \begin{enumerate}
 \item Give $\mathbf{Z}$ and $\mathbf{x}$ initial values.
 \item Repeat.
 \item Implement Algorithm \ref{G1} to allocate the powers.
 \item Solve problem (\ref{EHM6}) using Hungarian method to update $\mathbf{Z}$ with the given allocated power.
 \item For the updated powers and $\mathbf{Z}$, generate the matrix $S$ and determine the optimal link selection vector.
 \item Stop once there is no improvement in the objective function or the maximum number of the iterations is reached.
 \end{enumerate}
 \end{algorithm}



\subsection{Computational Complexity}

To best characterize the computational complexity of the proposed algorithm, we note first that the overall algorithm, i.e., Algorithm \ref{Algor2} solves three distinct problems sequentially. The problem of power allocation can be solved by implementing the derived closed-forms. In particular,  $3K$ equations are needed to be solved to find the variables  $P_j\ \forall j$, $P_i \forall i$, and $q_{i,j} \forall z_{i,j}=1$. For the user pairing problem, the computational complexity of Hungarian method is in the order of $O(K^3)$ \cite{grinman2015hungarian}. Finally, as shown earlier,  the link selection problem has a polynomial computational complexity $O(K\times K+K)$.  


\subsection{Baseline Approaches}
\subsubsection{Baseline 1 (NOMA approach)}
 The difference between NOMA and the proposed Co-NOMA scheme is that Co-NOMA allows the strong users to forward the weak users' signals through RF links (i.e., there is cooperation among users), which provides two options for the weak users, either to be served by the VLC AP or by the paired strong user through the hybrid VLC/RF link. In NOMA, however, the weak users have only one option, which is to be served through the direct VLC link (i.e., there is no cooperation among users in NOMA). NOMA can, therefore, be seen as a special case of the formulated problem, where the link selection vector $\mathbf{x}$ is set to zero throughout the optimization problem. Therefore, the optimal NOMA scheme can be found by allocating the power using the closed-form solutions (\ref{omga}) and (\ref{qij1}) for all possible user pairings. Under such approach, however, uncovered and blocked users would not be served.  It is important to note that if a strong user $j$ is paired with a blocked or uncovered user $i$ (i.e., $\Psi_i^{(w)}=0$), the power allocation for this pair is distributed as $P_j^{(s)}=q_{i,j}$, where $q_{i,j}$ is given by $q_{i,j}=\frac{w_{j}^{(s)}B_v}{2K\lambda}-\frac{1}{\Psi_j^{(s)}}$, which can be proven in a similar fashion as in proposition \ref{prop2}.
\subsubsection{Baseline 2} For the sake of additional algorithmic comparison, we also provide a simple solution from optimization perspective. Specifically, for pairing the users, we propose that the best strong user is paired with the worst weak one, the second best strong user is paired with the second worst weak user, and so on. The rationale behind adopting such approach as a benchmarking baseline is its capability to provide a relative fairness at a low computational complexity. For the link selection vector, this baseline chooses that each blocked or uncovered user (i.e., the users which have zero VLC channel) must be served through the relayed VLC/RF link, while the remaining weak users must be served through the direct VLC link. Under the above user pairing and link selection simple strategies, the power is then allocated using the derived closed-form solutions found in section \ref{PAWS}.
\begin{table}[!t]
\centering
\caption{Simulation Parameters}
\label{table1}
\begin{tabular}{|p{.45\textwidth} | p{.2\textwidth} |}
\hline
  Parameter Name& Parameter Value\\
 \hline

  Bandwidth of VLC AP, $B$ & $20$ MHz  \\

  The physical area of PDs, $A_{p}$ & $1$\ cm$^2$ \\
   Half-intensity radiation angle, $\theta_{1/2}$ & $60^o$\  \\
  Gain of optical filter, $g_{of}$ & $1$  \\
  Refractive index, $n$ & 1.5 \\
   Optical-to-electrical conversion factor, $\rho$& $0.53$ [A/W]\\
   Noise PSD of LiFi, $N_0$ & $10^{-21}$\ A$^2$/Hz  \\
  Maximum input bias current, $I_H$ & $600$ mA  \\

  Minimum input bias current, $I_L$ & $400$ mA  \\
  Fill factor, $f$ &0.75\\
  Electric-to-optical conversion factor, $\nu$ & 10 W/A\\

  Thermal voltage, $V_t$ & 25 mV \\

  Dark saturation current of the PD, $I_0$ & $10^{-10}$ A\\
  LED height, &$3$ m\\
  User height & $0.85$\\
  \hline

  RF    \\
  \hline
  The  breakpoint distance & 5 m\\
  Bandwidth & 16 MHz\\
  Central carrier frequency  & 2.4 GHz\\
  Angle of arrival/departure of LoS & 45$^o$\\
  Shadow fading standard deviation (before the breakpoint) & 3 dB  \\
  Shadow fading standard deviation (after the breakpoint) & 5 dB  \\
  PSD of the noise & -174 dBm/Hz\\
  \hline
\end{tabular}

 \end{table}

\section{Simulation Results}
This section evaluates the performance of the proposed hybrid VLC/RF Co-NOMA scheme and  the proposed solutions in terms of both sum-rate and system fairness. We investigate the FoV effects, number of users, blockage rate, and the cell size. Simulation parameters are given in Table \ref{table1}. The performance of the proposed algorithms is assessed through Monte-Carlo simulations, where every point in the numerical results is the average of implementing 1000 different user distributions within the given cell. The blockage rate is defined as the number of times that the user is blocked over the times of total simulation realizations. We use Jain's fairness index to measure the system fairness, which is given by
$\frac{\left(\sum_{j=1}^{K}\sum_{i=1}^{K}R_{i,j}\right)^2}{2K\left(\sum_{j=1}^{K}\sum_{i=1}^{K}R_{i,j}^2\right)}$.
\begin{figure}[!t]
\centering
\includegraphics[width=4in]{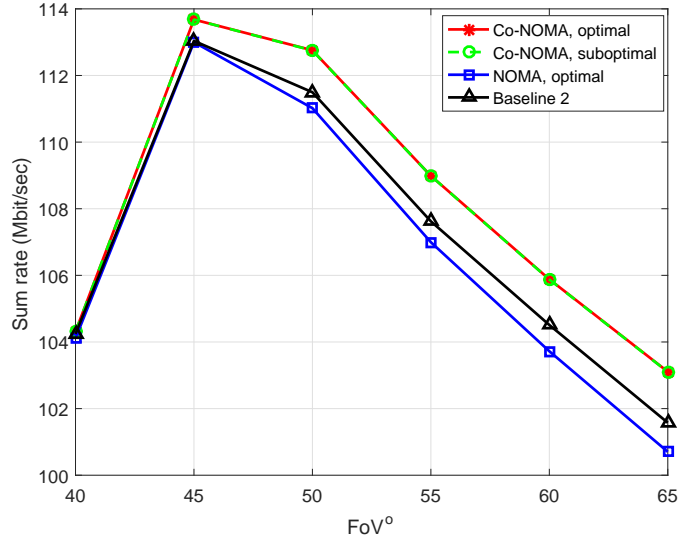}
\caption{Sum-rate versus users' FoV when number of users is $6$, the cell radius is $2.5$ m, and the blockage rate $0.1$.}
\label{SR_FoV}
\end{figure}

Fig. \ref{SR_FoV} compares the proposed hybrid VLC/RF Co-NOMA scheme and the exhaustive search with NOMA and the baseline approaches by plotting the sum-rate against the users' FoV. It can be seen from the figure that increasing the FoV of users leads to increasing the sum-rate, and then decreasing it for all approaches. Such behavior is due to the fact that the small users' FoV provides a potential for having some users to be uncovered, or to have have zero LoS channel gains. As the users' FoV first increases, the probability of coverage increases, which increases the sum-rate. But after some point, increasing the FoV would affect the
channel quality as illustrated through equations (\ref{vlcch}) and (\ref{FoVE}), which explains why the sum-rate decreases for larger values of the users' FoV.
\begin{figure}[!t]
\centering
\includegraphics[width=4in]{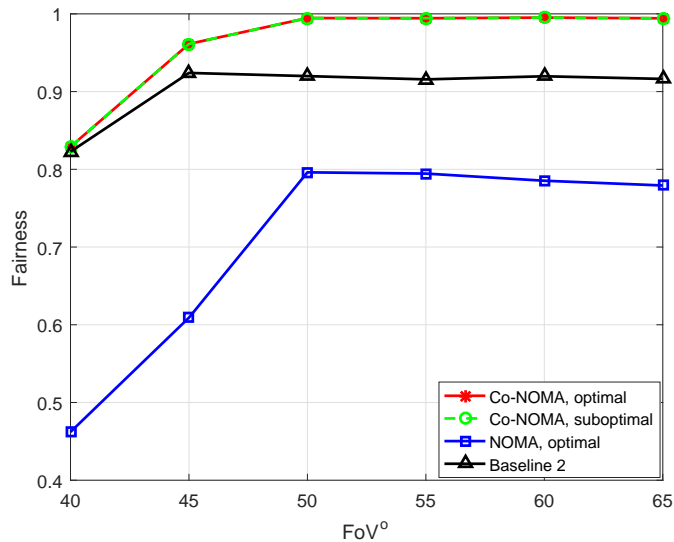}
\caption{System fairness versus users' FoV when number of users is $6$, the cell radius is $2.5$ m, and the blockage rate $0.1$.}
\label{F_FoV}
\end{figure}
Most importantly, the figure shows how the proposed Co-NOMA scheme outperforms all other algorithms for all values of the FoV, which highlights the important role of the proposed scheme in increasing the network throughput as compared to the classical NOMA scheme.

On the other hand, Fig. \ref{F_FoV} compares the proposed hybrid VLC/RF Co-NOMA scheme versus both NOMA and the proposed baseline approaches. Fig. \ref{F_FoV} plots Jain's fairness index versus the users' FoV. It can be seen that the fairness is low when the users' FoV is low, since the users which are far from AP would be out of the view (i.e., the LoS channel is zero), while the users that are close to the AP would get a good quality of service because of their channel quality. As the users' FoV increases, the probability that the number of covered users increases within a fixed certain area. Fig. \ref{F_FoV} particularly illustrates how the fairness of the proposed Co-NOMA scheme outperforms both NOMA scheme and the proposed baseline approach for all values of the FoV. In fact, both Fig.~\ref{SR_FoV} and Fig.~\ref{F_FoV} suggest that the proposed hybrid VLC/RF Co-NOMA scheme (optimal or suboptimal) outperforms  NOMA and the proposed baseline 2 in terms of both fairness and sum-rate. The improvement in terms of fairness is particularly pronounced, because  NOMA scheme cannot reach the out-of-coverage or blocked users, while the proposed hybrid VLC/RF Co-NOMA scheme can reach them through the hybrid VLC/RF relayed link. In addition, the hybrid VLC/RF links can provide the maximum fairness (rather than the direct VLC link) among the strong and weak users without affecting the sum-rate, as also illustrated earlier in Section \ref{PAWS}, Case 1.
\begin{figure}[!ht]
\centering
\includegraphics[width=4in]{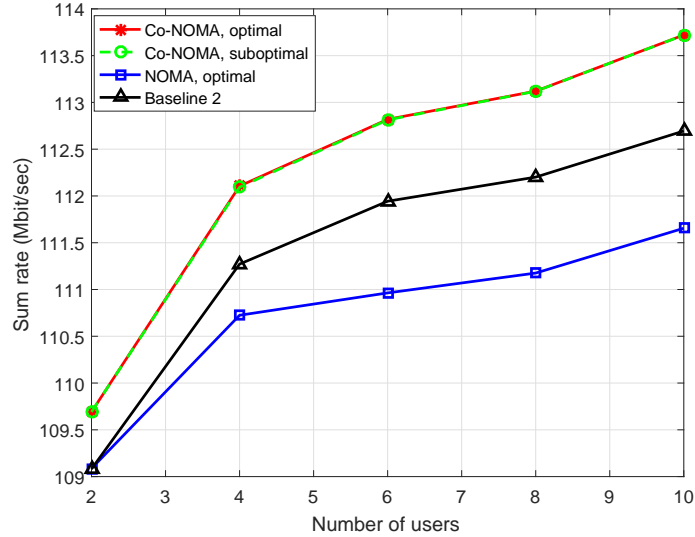}
\caption{Sum-rate versus number of users in the system when the cell radius is $2.5$ m, the blockage rate $0.1$, and the user FoV$=50^o$.}
\label{SR_NU}
\end{figure}

Fig. \ref{SR_NU} plots the sum-rate versus the total number of users located within a cell of radius $2.5$ m, while Fig. \ref{F_NU} plots the fairness of the same users and with the same cell size.
\begin{figure}[!ht]
\centering
\includegraphics[width=4in]{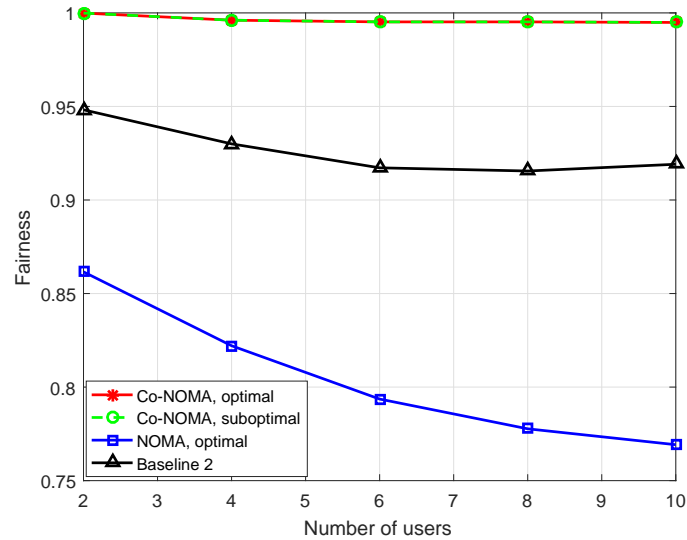}
\caption{System fairness versus number of users in the system when the cell radius is $2.5$ m, the blockage rate $0.1$, and the user FoV$=50^o$.}
\label{F_NU}
\end{figure}
In general, increasing the number of users in the system increases the sum-rate, but decreases the system fairness. However, this decrease in fairness (in Fig. \ref{F_NU}) is significant in the NOMA scheme and negligible in the proposed hybrid VLC/RF Co-NOMA scheme. On the other hand, the sum-rate in the Co-NOMA scheme increases in a faster rate than in NOMA. Figs. \ref{SR_NU} and \ref{F_NU} also show that the exhaustive search approach and the proposed Co-NOMA scheme provide the same performance, which is much better than the proposed baseline approaches.

\begin{figure}[!ht]
\centering
\includegraphics[width=4in]{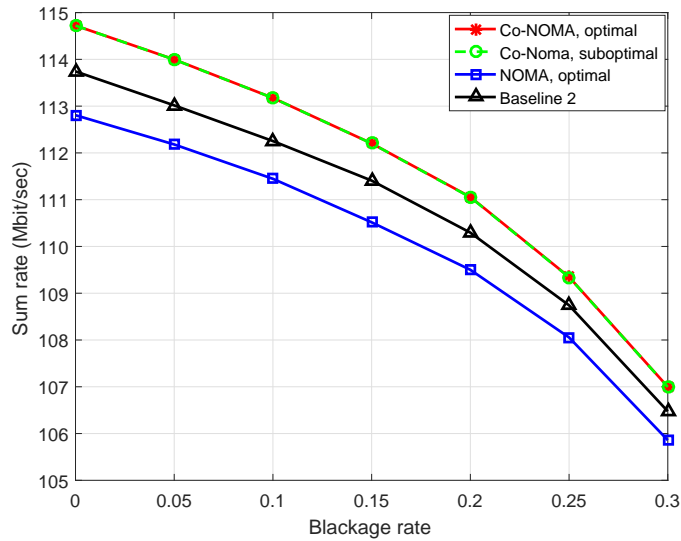}
\caption{Sum-rate versus blockage rate when, $N_u=6$, the cell radius is $2.5$ m,  and the users' FoV$=50^o$.}
\label{SR_BR}
\end{figure}

Fig. \ref{SR_BR} shows the effect of the blockage rate on the sum-rate. Increasing the blockage rate decreases the probability of the availability of the VLC LoS links to the users. In other words, the number of blocked users increases, which leads to decreasing the sum-rate of the system. Fig. \ref{SR_BR} particularly illustrates how the proposed hybrid VLC/RF Co-NOMA is better than NOMA for all given blockage rates, even when there is no blockage at all. This is because of the selection diversity at the weak user in Co-NOMA (the weak user in Co-NOMA can select the link that provides a maximum rate), while the weak user in NOMA has only one option to be served through (i.e., the VLC link).
\begin{figure}[!h]
\centering
\includegraphics[width=4in]{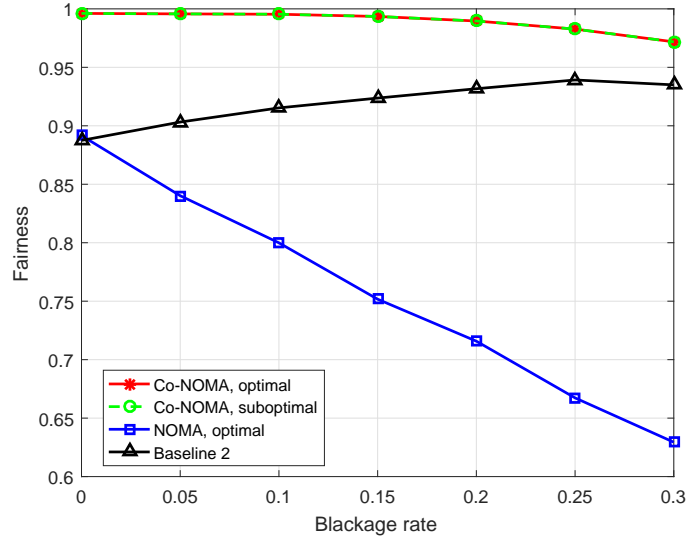}
\caption{System fairness versus blockage rate when, $N_u=6$, the cell radius is $2.5$ m, and the users' FoV$=50^o$.}
\label{F_BR}
\end{figure}

The effect of the blockage rate on the fairness is shown in Fig. \ref{F_BR}. The figure shows that increasing the blockage rate has a small impact on the fairness of the proposed hybrid VLC/RF Co-NOMA scheme until some point. This is because all the blocked users are considered as weak users and could be served through the paired strong users using the relayed link. But increasing the blockage rate further may result in having the number of blocked users greater than half of the total number of users, which affects also the fairness of the proposed Co-NOMA, albeit to a lesser degree than the impact shown on the NOMA fairness performance. The fairness of the baseline 2 increases with blocking rate since when there is no blockage, since all the users in the baseline 2 approach are served through the direct link, which results in fairness that is similar to NOMA. As the blockage rate increases, the number of served users through the relayed link increases, which results in approaching the maximum fairness reached by the proposed Co-NOMA scheme.

Fig. \ref{SR_CR} and Fig. \ref{F_CR}  show the impact of increasing the cell size on the sum-rate and fairness, respectively. As the cell size increases, the average channel quality decreases and the probability of having uncovered users increases. As a result, the sum-rate decreases as the cell size increases for Co-NOMA and NOMA and with different users' FoV.
\begin{figure}[!t]
\centering
\includegraphics[width=4in]{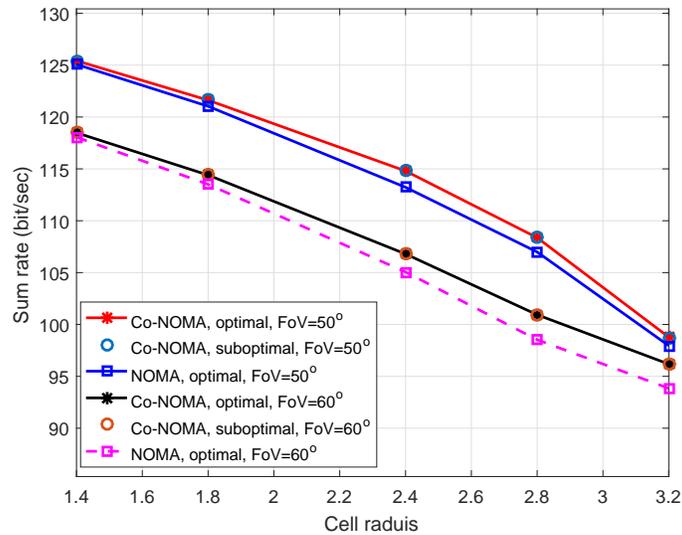}
\caption{Sum-rate versus the cell size when, $N_u=6$, blockage rate is 0.1,  and with different users' FoV.}
\label{SR_CR}
\end{figure}
\begin{figure}[!h]
\centering
\includegraphics[width=4in]{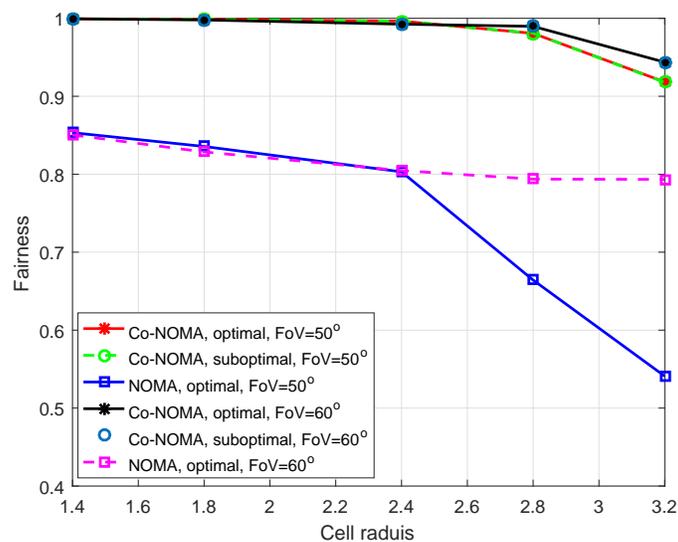}
\caption{System fairness versus the cell size when, $N_u=6$, blockage rate is 0.1,  and with different users' FoV.}
\label{F_CR}
\end{figure}
On the other hand, Fig. \ref{F_CR} shows that the fairness slightly decreases as the cell-radius increases for the proposed hybrid VLC/RF Co-NOMA scheme. This is because such scheme extends the coverage area by the RF link and increases the probability of coverage. In contrast, the NOMA scheme cannot reach the out-of-coverage users, which leads to having a high rate of reduction in system fairness. However, increasing the users' FoV would increase the coverage probability, but at the expense of decreasing the channel quality; thereby degrading the system sum-rate as shown in Fig. \ref{SR_CR}. Once again, Fig. \ref{SR_CR} and Fig. \ref{F_CR} highlight how the proposed Co-NOMA scheme outperforms the NOMA approach, both in terms of sum-rate and fairness, and for different cell sizes.

\section{Conclusion}
VLC is expected to be one of the candidate technologies in meeting the targeted requirements of next generation wireless communication networks. This paper introduces a novel cooperative scheme among users for extending coverage, improving sum-rate, and maximizing fairness in VLC systems. This cooperation is based on Co-NOMA, which can provide another chance for poorly serviced users to be served through a hybrid dual-hop VLC/RF link with the help of the well serviced users. The paper formulates an optimization problem that maximizes the weighted sum-rate by jointly allocating the power for users, pairing the users, and selecting the links for the weak users. An efficient, heuristic, iterative solution is proposed and compared with the exhaustive search approach, a simpler baseline solution, and with the traditional NOMA scheme. Simulation results show that a significant performance improvement in terms of sum-rate and fairness can be achieved by applying the proposed scheme and by jointly optimizing the system.

%
%
\appendices
\section{Proof of Proposition \ref{prop1}}
\label{AppA}
By writing the objective function in (\ref{EHM3}) as\\  \resizebox{0.6\hsize}{!}{$R_{i,j}=w_j^{(s)}\frac{B_v}{2K}\log_2(1+\Psi_j^{(s)} P_j^{(s)})+w_i^{(w)}\frac{B_v}{2K}\log_2(\frac{\Psi_j^{(s)} q_{i,j}+1}{\Psi_j^{(s)} P_j^{(s)}+1})$}, and by
  differentiating the function $R_{i,j}$ with respect to $P_j^{(s)}$, we obtain
\begin{equation}
\label{dR}
\frac{dR_{i,j}}{dP_j^{(s)}}=\frac{\Psi_j^{(s)} B_v(w_j^{(s)}-w_i^{(w)})}{2KP_j^{(s)}}.
\end{equation}
From (\ref{dR}), we can see that the objective function is an increasing function of $P_j^{(s)}$ if $w_j^{(s)}>w_i^{(w)}$, a decreasing function if $w_j^{(s)}<w_i^{(w)}$, and constant if $w_j^{(s)}=w_i^{(w)}$. This means that the sum-rate function is a constant function of $P_j^{(s)}$ and modifying the values of $w_i^{(w)}$ and $w_j^{(s)}$ just affects the weighted sum-rate but not the sum-rate itself. Hence, the maximum fairness can be implemented without any degradation in the sum-rate with setting $w_i^{(w)}=w_j^{(s)}$ and  having that the rate of the strong user equal to the rate of the weak user. To achieve that, $P_j^{(s)}$ must be selected to satisfy the following relation
\begin{equation}
\label{rb}
\frac{B_v}{2K}\log_2(1+\Psi_j^{(s)} P_j^{(s)})=\frac{B_v}{2K}\log_2(\frac{\Psi_j^{(s)} q_{i,j}+1}{\Psi_j^{(s)} P_j^{(s)}+1}).
\end{equation}
Solving (\ref{rb}), we obtain that $P_j^{(s)}=\eta_{i,j,1}$, where $\eta_{i,j,1}$ is given by
\begin{equation}
\label{eta1}
\eta_{i,j,1}=\frac{-1+\sqrt{1+q_{i,j}\Psi_j^{(s)} }}{\Psi_j^{(s)}},\  \ P_i^{(w)}=q_{i,j}-\eta_{i,j,1}.
\end{equation}
\section{Proof of Proposition \ref{prop2}}
\label{AppB}
It can be seen that the Hessian matrix of the objective function in (\ref{EHM5}) is negative definite whether $F_j^{(s)}(q_{i,j})$ and $F_i^{(w)}(q_{i,j})$ are given by (\ref{Fj1}) and (\ref{Fi1}), or given by (\ref{Fj2}) and (\ref{Fi2}). In addition, the constraints (\ref{EHM5}) are linear, which shows that problem (\ref{EHM5}) is convex. To find an optimal closed-form solution for $q_{i,j} \ \forall i,j$, write first the Lagrangian dual function:
\begin{multline}
 \resizebox{0.91\hsize}{!}{$\zeta= - \sum_{i=1}^K\sum_{j=1}^K x_iz_{i,j}\left(w_{j}^{(s)}F_j^{(s)}(q_{i,j})+w_{i}^{(w)}F_{i}^{(w)}(q_{i,j})\right)
-\sum_{i=1}^K\sum_{j=1}^K (1-x_i)z_{i,j}w_i^{(w)} \frac{B_v}{2K}\log_2(1+\Omega_{i,j}\Psi_j^{(s)})$}\\
 \resizebox{0.91\hsize}{!}{$-\sum_{i=1}^K\sum_{j=1}^K (1-x_i)z_{i,j}w_i^{(w)} \frac{B_v}{2K}\log_2\left(\frac{q_{i,j}\Psi_i^{(w)}+1}{\Omega_{i,j}\Psi_i^{(w)}+1}\right)+
\lambda\left(\sum_{i=1}^K\sum_{j=1}^K q_{i,j}-P_{max}\right)+\sum_{i=1}^K\sum_{j=1}^K \mu_{i,j}q_{i,j}$},
\end{multline}
where $\lambda$ is the dual variable associated with the sum-power constraint. Based on the first-order Karush-Kuhn-Tucker (KKT) conditions \cite{Boyd}, we have
\begin{equation}
\label{zet}
\frac{\partial \zeta}{\partial q_{i,j}}=0,\ \forall i,j.
\end{equation}
We have three cases. In the first case, if $z_{i,j}=0$ (i.e., users $i$ and $j$ are not paired), $q_{i,j}=0$ whether $x_i=1$ or $x_i=0$. The second case occurs when $z_{i,j}=1$ and $x_i=1$, and so we can reformulate (\ref{zet}) as
\begin{multline}
\label{qf}
\frac{\partial}{\partial q_{i,j}} \big[ -w_{j}^{(s)}F_j^{(s)}(q_{i,j})-w_{i}^{(w)}F_{i}^{(w)}(q_{i,j})+\lambda (q_{i,j}-P_{max})+
 \mu_{i,j}q_{i,j}\big]=0.
\end{multline}
If $F_j^{(s)}$ and $F_i^{(w)}$ are given by (\ref{Fj1}) and (\ref{Fi1}), respectively, (\ref{qf}) can be given by
\begin{multline}
\label{xe1}
\frac{\partial}{\partial q_{i,j}} \big[ -2w_{j}^{(w)}\frac{B_v}{2K}\log_2(\sqrt{\Psi_j^{(s)}q_{i,j}+1})+\lambda (q_{i,j}-P_{max}) 
+ \mu_{i,j}q_{i,j}\big]=0,
\end{multline}
where $w_i^{(w)}=w_j^{(s)}$ because $F_j^{(s)}=F_i^{(w)}$ in this case. On the other hand, if $F_j^{(s)}$ and $F_i^{(w)}$ are given by (\ref{Fj2}) and (\ref{Fi2}), respectively, (\ref{qf}) can be given by
\begin{multline}
\label{xe4}
\frac{\partial}{\partial q_{i,j}} \big[ -w_{j}^{(w)}\frac{B_v}{2K}\log_2(\frac{\Psi_j^{(s)}q_{i,j}+1}{A})-w_i^{(w)}R_{i,j}^{RF}(\mathbf{x})
+\lambda (q_{i,j}-P_{max})+ \mu_{i,j}q_{i,j}\big]=0,
\end{multline}
Solving (\ref{xe1}) or (\ref{xe4}), we obtain the same expression for $q_{i,j}$, which is given by
\begin{equation}
\label{qij2}
q_{i,j}=\big[\frac{w_{j}^{(w)}B_v}{2K\lambda}-\frac{1}{\Psi_j^{(s)}}\big]^+.
\end{equation}
Finally, in the third case, if we have $z_{i,j}=1$ and $x_i=0$ (the weak user is served through the direct VLC link), which allows to rewrite the first-order condition (\ref{zet}) as follows:
\begin{multline}
\label{xe2}
\frac{\partial}{\partial q_{i,j}} \bigg[ -w_{i}^{(w)}\frac{B_v}{2K}\log_2(\frac{q_{i,j}\Psi_i+1}{\Omega_{i,j}\Psi_i^{(w)}+1})+\lambda (q_{i,j}-P_{max})
+ \mu_{i,j}q_{i,j}\bigg]=0
\end{multline}
Solving (\ref{xe2}), we get
\begin{equation}
\label{qij3}
q_{i,j}=\bigg[\frac{w_{i}^{(w)}B_v}{2K\lambda}-\frac{1}{\Psi_j^{(s)}}\bigg]^+.
\end{equation}
\section*{Acknowledgment}
The authors acknowledge funding from the Research and Development (R\&D) Program (Research Pooling Initiative), Ministry of Education, Riyadh, Saudi Arabia.

\bibliography{mylib}

\bibliographystyle{IEEEtran}

\end{document}